\documentclass[11pt, a4paper, reqno]{amsart}

\usepackage{anysize}
\usepackage[utf8]{inputenc}
\usepackage{amsmath}
\usepackage{mathrsfs} 
\usepackage{amsthm}
\usepackage{mathtools}
\usepackage{thmtools}
\usepackage{thm-restate}
\usepackage{amssymb}
\usepackage{enumerate}
\usepackage{tcolorbox}
\usepackage{bm}
\usepackage{complexity}

\usepackage{hyperref}

\marginsize{2.5cm}{2.5cm}{3cm}{3cm}

\newtheorem{theorem}{Theorem}
\numberwithin{theorem}{section}
\newtheorem{corollary}[theorem]{Corollary}
\newtheorem{proposition}[theorem]{Proposition}

\newtheorem{lemma}[theorem]{Lemma}
\numberwithin{lemma}{section}
\newtheorem{observation}[theorem]{Observation}
\numberwithin{observation}{section}

\newtheorem{innercustomthm}{Theorem}
\newenvironment{customthm}[1]
  {\renewcommand\theinnercustomthm{#1}\innercustomthm}
  {\endinnercustomthm}
  
\newtheorem{innercustomcor}{Corollary}
\newenvironment{customcor}[1]
  {\renewcommand\theinnercustomcor{#1}\innercustomcor}
  {\endinnercustomcor}
  
\newtheorem{innercustomconj}{Conjecture}
\newenvironment{customconj}[1]
  {\renewcommand\theinnercustomconj{#1}\innercustomconj}
  {\endinnercustomconj}

\theoremstyle{definition}
\newtheorem{definition}{Definition}
\numberwithin{definition}{section}

\AfterEndEnvironment{definition}{\noindent\ignorespaces}
\AfterEndEnvironment{lemma}{\noindent\ignorespaces}
\AfterEndEnvironment{proof}{\noindent\ignorespaces}
\AfterEndEnvironment{corollary}{\noindent\ignorespaces}
\AfterEndEnvironment{proposition}{\noindent\ignorespaces}
\AfterEndEnvironment{theorem}{\noindent\ignorespaces}
\AfterEndEnvironment{innercustomthm}{\noindent\ignorespaces}

\newcommand{\NN}{\mathbb{N}}			
\newcommand{\RR}{\mathscr{R}}			
\renewcommand{\PPP}{\mathscr{P}}			

\renewcommand{\C}{\mathcal{C}}			
\renewcommand{\L}{\mathcal{L}}			
\newcommand{\words}[1]{\langle #1 \rangle}	

\newcommand{\f}{\mathsf{f}}				
\renewcommand{\poly}{\mathsf{poly}}		
\newcommand{\quasipoly}{\mathsf{qpoly}}	
\renewcommand{\exp}{\mathsf{exp}}			
\newcommand{\dtime}{\mathtt{dtime}}		
\newcommand{\ntime}{\mathtt{ntime}}		
\newcommand{\dspace}{\mathtt{dspace}}		

\renewcommand{\W}{\mathsf{W}}
\newcommand{\U}{\mathsf{U}}

\newcommand{\defeq}{\coloneqq}			
\newcommand{\eref}[1]{Eq.~(\ref{#1})}		
\newcommand{\sref}[1]{Section~\ref{#1}}		
\newcommand{\lemref}[1]{Lemma~\ref{#1}}	
\newcommand{\propref}[1]{Proposition~\ref{#1}}	
\newcommand{\thmref}[1]{Theorem~\ref{#1}}	
\newcommand{\corref}[1]{Corollary~\ref{#1}}	
\newcommand{\conjref}[1]{Conjecture~\ref{#1}}
\newcommand{\obsref}[1]{Observation~\ref{#1}}
\newcommand{\RTIME}{\mathsf{RTIME}}
\newcommand{\UTIME}{\mathsf{UTIME}}
\newcommand{\PTIME}{\mathsf{PTIME}}
\renewcommand{\ae}{\ \mathrm{a.e.}}

\newcommand{\id}{\mathrm{id}}

\begin{document}

\title[A Refinement of the McCreight-Meyer Union Theorem]{A Refinement of the McCreight-Meyer Union Theorem}
\author{Matthew Fox}
\address{Department of Physics, University of Colorado, Boulder, CO}
\email{matthew.fox@colorado.edu}

\author{Chaitanya Karamchedu}
\address{Department of Computer Science, University of Maryland, College Park, MD}
\email{cdkaram@umd.edu}

\begin{abstract}



Using properties of Blum complexity measures and certain complexity class operators, we exhibit a total computable and non-decreasing function $t_\poly$ such that for all $k$, $\Sigma_k\P = \Sigma_k\TIME(t_\poly)$, $\BPP = \BPTIME(t_\poly)$, $\RP = \RTIME(t_\poly)$, $\UP = \UTIME(t_\poly)$, $\PP = \PTIME(t_\poly)$, $\Mod_k\P = \Mod_k\TIME(t_\poly)$, $\PSPACE = \DSPACE(t_\poly)$, and so forth. A similar statement holds for any collection of language classes, provided that each class is definable by applying a certain complexity class operator to some Blum complexity class.

\end{abstract}

\maketitle

\section{Introduction}
\label{sec:intro}

Just two years after Blum \cite{blum} devised his eponymous axioms on complexity measures, McCreight and Meyer \cite{mccreightmeyer} proved that every such measure satisfies a ``union theorem''. For deterministic time, non-deterministic time, and deterministic space, all of which are Blum measures, the McCreight-Meyer Union Theorem exhibits three total computable (but non-constructible) functions $t_1, t_2$, and $t_3$ such that $\P = \DTIME(t_1)$, $\NP = \NTIME(t_2)$, and $\PSPACE = \DSPACE(t_3)$. However, since McCreight and Meyer's proof sensitively depends on the complexity measure in question, it is not at all obvious how $t_1$, $t_2$, and $t_3$ relate.

In this paper, we refine McCreight and Meyer's result and exhibit a total computable and non-decreasing function $t_\poly$ such that for all $k$,
\begin{align*}
\,\,\,\Sigma_k\P = \Sigma_k\TIME(t_\poly), \quad \BPP = \BPTIME(t_\poly), \quad \RP = \RTIME(t_\poly),\quad \UP = \UTIME(t_\poly),  \quad \quad \;\;\;\\
\PP = \PTIME(t_\poly), \quad \Mod_k\P = \Mod_k\TIME(t_\poly), \quad \text{and} \quad \PSPACE = \DSPACE(t_\poly).\quad\quad\quad\quad
\end{align*}
Apparently, $t_\poly$ is inextricably tied to the polynomial-boundedness of all these classes while simultaneously agnostic to the exact underlying computational model. Unfortunately, this suggests that a better understanding of $t_\poly$ will probably not help resolve the outstanding relationships between these classes.


The above result follows in two elementary steps. First, we show that for any finite collection of Blum measures $\Phi_1, \dots, \Phi_\ell$, which define Blum complexity classes $\C_{\Phi_1}, \dots, \C_{\Phi_\ell}$, respectively, there is a total computable and non-decreasing function $t_\poly$ such that for all $i \in [\ell]$, $\C_{\Phi_i}(t_\poly) = \bigcup_{f \in \poly} \C_{\Phi_i}(f),$ where $\poly$ is the set of polynomial functions. Among other things, this establishes that $\P = \DTIME(t_\poly)$, $\NP = \NTIME(t_\poly)$, and $\PSPACE = \DSPACE(t_\poly)$. Next, we exploit the fact that classes like $\Sigma_k\P, \BPP, \UP$ and so forth are definable by applying a particular complexity class operator to $\P$. For example, $\Sigma_k\P = \Sigma_k^\poly \cdot \P$. Thus, by the first step, $\Sigma_k\P = \Sigma_k^\poly \cdot \DTIME(t_\poly)$, and a final elementary argument, which crucially relies on the fact that $t_\poly$ is non-decreasing, establishes $\Sigma_k^\poly \cdot \DTIME(t_\poly) = \Sigma_k\TIME(t_\poly)$, as desired.

A more intricate argument works for any (potentially infinite) collection of language classes, provided that each class in the collection is definable by applying a certain complexity class operator to some Blum complexity class. This more general result is the context of our main theorem. We regard it as evidence for a plausible conjecture that was raised to us by Fortnow:

\begin{customconj}{A}[Fortnow \cite{fortnow3}]
\label{conj:mainconj}
Let $\{\Phi_i \mid i \in \NN \}$ be a partial recursive list of Blum measures and let $\f = \{f_i \mid i \in \NN\}$ be a total recursive list of functions such that if $i \leq j$, then $f_i(n) \leq f_j(n)$ for almost all $n$. There exists a total computable function $t_\f$ such that for all $i$, $\C_{\Phi_i}(t_\f) = \bigcup_{f \in \f} \C_{\Phi_i}(f).$
\end{customconj}

\section{Preliminaries}

In this paper, $\Sigma$ is a non-unary alphabet, $\Sigma^n$ is the set of all $\Sigma$-strings with length $n$, $\Sigma^*$ is the set of all $\Sigma$-strings with finite length, $|x|$ is the length of the $\Sigma$-string $x$, $[\ell]$ is the set $\{1,2,\dots, \ell\}$, $\#X$ is the cardinality of the set $X$, $\id : n \mapsto n$ is the identity function on the non-negative integers $\NN$, $\chi_L : \Sigma^* \rightarrow \{0,1\}$ is the characteristic function of the language $L \subseteq \Sigma^*$, $\equiv_k$ is congruence modulo $k$, $\PPP(X_1 \times \dots \times X_k \rightarrow Y)$ is the set of all $k$-ary partial computable (p.c.) functions from $X_1 \times \cdots \times X_k$ to $Y$, where each $X_i$ and $Y$ is either $\Sigma^*$, $\NN$, or $\{0,1\}$, and $\RR(X_1 \times \dots \times X_k \rightarrow Y)$ is the set of all $k$-ary total computable functions from $X_1 \times \cdots \times X_k$ to $Y$. Given $f \in \PPP(X_1 \times \dots \times X_k \rightarrow Y)$, we write $f(x_1, \dots, x_k)\downarrow$ if and only if (iff) $(x_1, \dots, x_k)$ is in the domain of $f$.

We say $f \in \RR(\NN \rightarrow \NN)$ is \emph{non-decreasing} iff $m \leq n$ implies $f(m) \leq f(n)$. Given $f, g \in \RR(\NN \rightarrow \NN)$, by a statement like ``$f \leq g$'', we mean ``$f(n) \leq g(n)$ almost everywhere (a.e.)'', i.e., ``$f(n) \leq g(n)$ for all but finitely many $n$''. Similarly, given $f \in \PPP(\Sigma^* \rightarrow \NN)$, $g \in \PPP(\NN \times \Sigma^* \rightarrow \NN)$, and $h \in \RR(\NN \rightarrow \NN)$, by statements like ``$f(x) \leq h(|x|) \ae$'' and ``$g(e, x) \leq h(|x|) \ae$'', we mean ``$f(x) \leq h(|x|)$ for all but finitely many $x$'' and ``for this particular choice of $e$, $g(e, x) \leq h(|x|)$ for all but finitely many $x$'', respectively. 

We employ any standard pairing function $\words{\cdot, \cdot} \in \RR(\Sigma^* \times \Sigma^* \rightarrow \Sigma^*)$ that satisfies $|x_1| + |x_2| \leq |\words{x_1, x_2}| \leq |x_1| + |x_2| + O(\log |x_1|)$. Of course, this induces a $k$-ary pairing function given by 
$$
\words{x_1, \dots, x_k} \defeq \words{x_1, \words{x_2, \dots, \words{x_{k-1}, x_k}}},
$$
which satisfies
\begin{equation}
\sum_{i = 1}^k |x_i| \leq |\words{x_1, \dots, x_k}| \leq \sum_{i = 1}^k |x_i| + O\left(\sum_{i = 1}^{k - 1} \log |x_i|\right).
\label{eq:pairingfunctionsize}
\end{equation}
Also, if $\epsilon$ is the empty string in $\Sigma^*$, then every string $x$ gets encoded as $\words{x} \defeq \words{x, \epsilon}$, which satisfies $|x| \leq |\words{x}| \leq |x| + O(\log |x|).$


\subsection{Acceptable Collections} Our results apply only for particular sets of functions which we call \emph{acceptable collections}.
\begin{definition}
A countable set of functions $\f = \{f_i \mid i \in \NN\}$ is an \emph{acceptable collection} iff:
\begin{enumerate}[(a)]
\item $\f$ is total recursive, i.e., there is $F \in \mathscr{R}(\NN \times \NN \rightarrow \NN)$ such that for all $i$ and $n$, $F(i,n) = f_i(n)$,\footnote{Therefore, $f_i \in \RR(\NN \rightarrow \NN)$ for every $i$.}
\item for all $i$ and $j$, $i \leq j$ implies $f_{i} \leq f_{j}$,
\item for all $i$, $f_i$ is non-decreasing,
\item there exists $i$ such that $f_i$ is unbounded,
\item for all $i, j_1, \dots, j_\ell$, there exists $k$ such that 
\begin{equation}
f_i \circ \left(\id + \sum_{m = 1}^\ell f_{j_m} + O\left(\log + \sum_{m = 1}^{\ell - 1} \log \circ f_{j_m}\right)\right) \leq f_k.
\label{eq:property6}
\end{equation}
\end{enumerate}
\end{definition}

Examples of acceptable collections include the logarithmic functions $\mathsf{log} \defeq \{i\log(n) \mid i \in \NN\}$, the polylogarithmic functions $\polylog \defeq \{\log^i(n) \mid i \in \NN\}$, the polynomial functions $\poly \defeq \{n^i \mid i \in \NN\}$, and the quasi-polynomial functions $\quasipoly \defeq \{2^{\log^i(n)} \mid i \in \NN\}$. Thus, many natural sets of functions are acceptable collections. That said, the exponential functions $\exp \defeq \{2^{n^i} \mid i \in \NN\}$ are \emph{not} an acceptable collection as they violate property (e). It remains open, therefore, if our results apply to $\exp$.

Altogether, properties (a) and (b) are necessary for the McCreight-Meyer Union Theorem to hold (\sref{sec:MMUnionTheorem}). The additional properties (c), (d), and (e) are used to prove our main results, but it is unclear if they are actually necessary. We note that the unusual property (e) is ultimately forced by our choice of pairing function (c.f. Eqs.~(\ref{eq:pairingfunctionsize}) and (\ref{eq:property6})).

\subsection{The Blum Axioms} We now recall the Blum axioms \cite{blum}.
\begin{definition}
Let $\{\varphi_e \mid e \in \NN\}$ be an acceptable G\"odel numbering of $\PPP(\Sigma^* \rightarrow \Sigma^*)$. We call $\Phi \in \PPP(\NN \times \Sigma^* \rightarrow \NN)$ a \emph{Blum complexity measure} (or \emph{Blum measure} for short) iff:
\begin{enumerate}[(i)]
\item for all $e$ and $x$, $\Phi(e, x)\downarrow$ iff $\varphi_e(x)\downarrow$,
\item there exists $f \in \RR(\NN \times \Sigma^* \times \NN \rightarrow \NN)$ such that for all $e$, $x$, and $m$,
$$
f(e,x,m) = 
\begin{cases}
1 &\text{if $\Phi(e, x) = m$},\\
0 &\text{otherwise}.
\end{cases}
$$
\end{enumerate}
Together, every $t \in \RR(\NN \rightarrow \NN)$ and Blum measure $\Phi$ define a \emph{Blum complexity class} (or \emph{Blum class} for short) given by
$$
\C_{\Phi}(t) \defeq \left\{\varphi_e \mid \Phi(e, x) \leq t(|x|)\ae\right\}.
$$
\end{definition}

For example, the number of steps single-tape, deterministic Turing machines take on a given input defines the Blum measure $\dtime \in \PPP(\NN \times \Sigma^* \rightarrow \NN)$, which stands for \emph{deterministic time}. In particular, for the $e$th p.c. function $\varphi_e$, $\dtime$ is such that for every $x$, $\dtime(e, x) = m$ iff there exists a deterministic Turing machine $M$ such that $M(x) = \varphi_e(x)$ and $M$ halts in exactly $m$ steps. Other examples include $\dspace$ (deterministic space) and $\ntime$ (non-deterministic time), which are defined like $\dtime$ but with the obvious modification.

Besides a few exceptional statements, in this paper we shall mostly reason at the level of language classes. Indeed, every Blum class $\C_{\Phi}$ defines a \emph{Blum language class} $\L_{\Phi}$, which is got by considering those languages whose characteristic functions are in $\C_{\Phi}$:
$$
\L_{\Phi}(t) \defeq \left\{L \subseteq \Sigma^* \mid \chi_L \in \C_{\Phi}(t)\right\}.
$$
For example, $\L_\dtime(t) = \DTIME(t)$, $\L_\ntime(t) = \NTIME(t)$, and $\L_\dspace(t) = \DSPACE(t)$.

\subsection{The McCreight-Meyer Union Theorem}
\label{sec:MMUnionTheorem}

In addition to satisfying a speedup theorem \cite{blum}, a compression theorem \cite{blum}, and a gap theorem \cite{borodin, trakhtenbrot}, every Blum measure also satisfies a union theorem.

\begin{theorem}[McCreight-Meyer Union Theorem \cite{mccreightmeyer}]
\label{thm:MMTheorem}
Let $\Phi$ be a Blum measure and let $\f = \{f_i \mid i \in \NN\}$ be a collection of functions that at least satisfies conditions (a) - (d) in the definition of an acceptable collection. There exists $t_\f \in \RR(\NN \rightarrow \NN)$ such that:
\begin{enumerate}[(i)]
\item $t_\f$ is non-decreasing,
\item for all $f \in \f$, $f \leq t_\f$,
\item $\C_\Phi(t_\f) = \bigcup_{f \in \f} \C_\Phi(f).$
\end{enumerate}
\end{theorem}

If $\f$ only satisfies conditions (a) and (b) in the definition of an acceptable collection, then there still exists $t_\f \in \RR(\NN \rightarrow \NN)$ such that statements (ii) and (iii) hold in \thmref{thm:MMTheorem}. This fact is the usual statement of the McCreight-Meyer Union Theorem (Theorem 5.5 in \cite{mccreightmeyer}, see also Theorem J.3 in \cite{kozen} for a more modern treatment). However, in this paper it is vital that $t_\f$ is non-decreasing, which is the case if $\f$ also satisfies conditions (c) and (d) in the definition of an acceptable collection. For more on this, we refer the reader to Remark 5.13 in \cite{mccreightmeyer}.

A fascinating corollary of \thmref{thm:MMTheorem} is the following.

\begin{corollary}
\label{cor:intcorollary1}
There exist non-decreasing $t_1, t_2, t_3 \in \RR(\NN \rightarrow \NN)$ such that:
\begin{enumerate}[(i)]
\item $\P = \DTIME(t_1)$,
\item $\NP = \NTIME(t_2)$,
\item $\PSPACE = \DSPACE(t_3)$.
\end{enumerate}
\end{corollary}

We now prove a proposition that has many interesting corollaries.

\begin{proposition}
Let $\{\Phi_i \mid i \in [\ell] \}$ be a finite set of Blum measures and let $\f$ be a collection of functions that at least satisfies conditions (a) - (d) in the definition of an acceptable collection. There exists $t_\f \in \RR(\NN \rightarrow \NN)$ such that:
\begin{enumerate}[(i)]
\item $t_\f$ is non-decreasing,
\item for all $f \in \f$, $f \leq t_\f$,
\item for all $i \in [\ell]$, $\C_{\Phi_i}(t_\f) = \bigcup_{f \in \f}\C_{\Phi_i}(f)$.
\end{enumerate}
\label{prop:mainprop1}
\end{proposition}

\begin{proof}
By \thmref{thm:MMTheorem}, for all $i \in [\ell]$, there is a non-decreasing function $t_{\f, i} \in \RR(\NN \rightarrow \NN)$ such that $\C_{\Phi_i}(t_{\f, i}) = \bigcup_{f \in \f} \C_{\Phi_i}(f)$ and for all $f \in \f$, $f \leq t_{\f,i}$. For every $n$, define $t_\f(n) \defeq \min_{i \in [\ell]} t_{\f, i}(n)$. Evidently, $t_\f$ is total computable because the map $(i,n) \mapsto t_{\f, i}(n)$ is, and $t_\f$ is non-decreasing because every $t_{\f,i}$ is. Additionally, for all $i \in [\ell]$, $t_\f(n) \leq t_{\f, i}(n)$ for all $n$. Therefore, $\C_{\Phi_i}(t_\f) \subseteq \C_{\Phi_i}(t_{\f, i}).$ Conversely, for all $i \in [\ell]$ and all $f \in \f$, $f \leq t_{\f,i}$, i.e., there is $N_{f,i} \in \NN$ such that $f(n) \leq t_{\f, i}(n)$ for all $n > N_{f,i}$. Therefore, $f(n) \leq \min_{i \in [\ell]} t_{\f, i}(n) = t_\f(n)$ for all $n > \max_{i \in [\ell]} N_{f,i}$. Consequently, $f \leq t_\f$, so also $\C_{\Phi_i}(t_{\f, i}) \subseteq \C_{\Phi_i}(t_\f)$.
\end{proof}

An immediate corollary of \propref{prop:mainprop1} is the following, which is a considerably stronger version of \corref{cor:intcorollary1}.

\begin{corollary}
\label{cor:intcorollary}
There exists a non-decreasing function $t_\poly \in \RR(\NN \rightarrow \NN)$ such that:
\begin{enumerate}[(i)]
\item $\P = \DTIME(t_\poly)$,
\item $\NP = \NTIME(t_\poly)$,
\item $\PSPACE = \DSPACE(t_\poly)$.
\end{enumerate}
\end{corollary}

Ultimately, our goal in this paper is to extend \propref{prop:mainprop1} to an infinity of classes---such as every level of the polynomial hierarchy---which is our contribution toward Fortnow's \conjref{conj:mainconj}, which we restate below for convenience.

\begin{customconj}{A}[Fortnow \cite{fortnow3}]
\label{conj:mainconj2}
Let $\{\Phi_i \mid i \in \NN \}$ be a partial recursive list of Blum measures and let $\f$ be a collection of functions that at least satisfies conditions (i) and (ii) in the definition of an acceptable collection. There exists $t_\f \in \RR(\NN \rightarrow \NN)$ such that for all $i$, $\C_{\Phi_i}(t_\f) = \bigcup_{f \in \f} \C_{\Phi_i}(f).$
\end{customconj}

Towards this conjecture, it is natural to try to generalize our proof of \propref{prop:mainprop1} to the infinite case by taking $t_\f(n) \defeq \min_{i \leq n}t_{\f, i}(n)$ for every $n$, which is in fact total computable. Whereas for every $i$, $t_\f \leq t_{\f, i}$ so that $\C_{\Phi_i}(t_\f) \subseteq \bigcup_{f \in \f} \C_{\Phi_i}(f)$, it is not obviously the case that $\C_{\Phi_i}(t_\f) \supseteq \bigcup_{f \in \f} \C_{\Phi_i}(f)$. For example, it is not obviously the case that $f \leq t_\f$ for every $f \in \f$. Therefore, while establishing this latter containment in the finite case is relatively straightforward, it is not so---and in fact it appears to be rather difficult---in the infinite case.

\section{Complexity Class Operators for General Blum Language Classes}

Our main result concerns language classes that are definable by applying a certain complexity class operator to some Blum language class. While there are many different complexity class operators, the bulk of them operate either (i) by alternately existentially and universally quantifying strings of a certain size that satisfy some global condition or (ii) by counting ``witness strings'' of a certain size that also satisfy some global condition. This observation compels us to introduce two new complexity class operators, which collectively encapsulate many of the well-known operators.

\subsection{The Operator $\Sigma_{k}^\f$ and the Class $\Sigma_k\L^\f_\Phi$} The levels $\Sigma_k\P$ of the polynomial hierarchy exhibit certain structure that not all language classes afford. In particular, each level $\Sigma_k\P$ is definable as the class obtained after applying the complexity class operator $\Sigma_k^\poly$ to the Blum language class $\P$. We define the most general form of this operator as follows.

\begin{definition}
Let $\f$ be an acceptable collection and let $k \in \NN$. For a class $\L$ of languages, $\Sigma_{k}^\f \cdot \L$ is the class of languages $L$ for which there exist $f_1, \dots, f_k \in \f$ and $L' \in \L$ such that for all $x$,
\begin{equation}
x \in L \iff \left(\exists y_1 \in \Sigma^{f_1(|x|)}\right)\left(\forall y_2 \in \Sigma^{f_2(|x|)}\right)\dots\left(Q_k y_k \in \Sigma^{f_k(|x|)}\right) \chi_{L'}(\words{x, y_1, \dots, y_k}) = 1,
\label{eq:operatoraction}
\end{equation}
where $Q_k$ is $\exists$ if $k$ is odd and $\forall$ if $k$ is even.
\end{definition}

It is easily seen that $\Sigma_k^\f$ obeys the following monotonicity property.

\begin{observation}
For all acceptable collections $\f$, all $k$, and all language classes $\L_1$ and $\L_2$, $\L_1 \subseteq \L_2$ implies $\Sigma_k^\f \cdot \L_1 \subseteq \Sigma_k^\f \cdot \L_2$.
\label{obs:monotone}
\end{observation}

Importantly, for general $t$, the class $\Sigma_k^\poly \cdot \DTIME(t)$ is not $\Sigma_k\TIME(t)$, which is defined as follows.

\begin{definition}
The class $\Sigma_k\TIME(t)$ consists of all languages $L$ that are decided by alternating Turing machines that begin in the existential state, alternate at most $k - 1$ times, and halt in at most $t$ steps. Equivalently, $L \in \Sigma_k\TIME(t)$ iff there exist $f_1, \dots, f_k \in \poly$ and $\varphi_e \in \PPP(\Sigma^* \rightarrow \{0,1\})$ such that for all $x, w_1, \dots, w_k$:
\begin{enumerate}[(i)]
\item $\dtime(e, \words{x, w_1, \dots, w_k}) \leq t(|\words{x}|)$,
\item and
$$
x \in L \iff \left(\exists y_1 \in \Sigma^{f_1(|x|)}\right)\left(\forall y_2 \in \Sigma^{f_2(|x|)}\right)\dots\left(Q_k y_k \in \Sigma^{f_k(|x|)}\right) \varphi_e(\words{x, y_1, \dots, y_k}) = 1.
$$
\end{enumerate}
\end{definition}

Here, condition (i) enforces that $\varphi_e$ is computable on a deterministic Turing machine that halts in at most $t(|\words{x}|)$ steps, as opposed to, say, $t(|\words{x, y_1, \dots, y_k}|)$ steps. Thus, in contrast to the definition of the class $\Sigma_k^\poly \cdot \DTIME(t)$, languages in $\Sigma_k\TIME(t)$ are decided by deterministic Turing machines whose runtimes are \emph{not} influenced by the witness strings $y_1, \dots, y_k$. Consequently, for general $t$, it follows from the time hierarchy theorem that $\Sigma_k^\poly \cdot \DTIME(t) \neq \Sigma_k\TIME(t)$.

We now introduce the class $\Sigma_k\L_\Phi^\f(t)$, which is defined analogically to the alternate but equivalent definition of $\Sigma_k\TIME(t)$, but for any acceptable collection $\f$ and Blum measure $\Phi$.

\begin{definition}
Let $\f$ be an acceptable collection and let $\Phi$ be a Blum measure. The class $\Sigma_k\L_\Phi^\f(t)$ consists of all languages $L$ for which there exist $f_1, \dots, f_k \in \f$ and $\varphi_e \in \PPP(\Sigma^* \rightarrow \{0,1\})$ such that for all $x, w_1, \dots, w_k$:
\begin{enumerate}[(i)]
\item $\Phi(e, \words{x, w_1, \dots, w_k}) \leq t(|\words{x}|)$,
\item and
$$
x \in L \iff \left(\exists y_1 \in \Sigma^{f_1(|x|)}\right)\left(\forall y_2 \in \Sigma^{f_2(|x|)}\right)\dots\left(Q_k y_k \in \Sigma^{f_k(|x|)}\right) \varphi_e(\words{x, y_1, \dots, y_k}) = 1.
$$
\end{enumerate}
\end{definition}

As above, here condition (i) enforces that $\varphi_e$ is computable in less than $t(|\words{x}|)$ of the Blum measure $\Phi$. In other words, the amount of $\Phi$ it takes to compute $\varphi_e$ is \emph{independent} of the witness strings $y_1, \dots, y_k$. Again, this is in contrast to the definition of the class $\Sigma_k^\f \cdot \L_\Phi(t)$, where the witness strings \emph{do} implicitly influence the amount of $\Phi$ it takes to compute $\chi_{L'}$ in \eref{eq:operatoraction}. For general $t$, therefore, we expect $\Sigma_k^\f \cdot \L_\Phi(t) \neq \Sigma_k\L_\Phi^\f(t)$. Nevertheless, as we prove in \sref{sec:mainproofs}, if $t = t_\f$ in \propref{prop:mainprop1}, then equality does in fact hold.

These same considerations hold for the complement operator $\Pi_k^\f$ and the class $\Pi_k\L^\f_\Phi(t)$, which are got from $\Pi_k^\poly$ and $\Pi_k\TIME(t)$, respectively, in the same manner that $\Sigma_k^\f$ and $\Sigma_k\L^\f_\Phi(t)$ were got from $\Sigma_k^\poly$ and $\Sigma_k\TIME(t)$, respectively.

\subsection{The Operator $\W_\omega^\f$ and the Class $\W_\omega\L_\Phi^\f$}

Like $\Sigma_k\P$, classes like $\BPP$, $\RP$, and $\UP$ are definable by applying a particular complexity class operator to the Blum language class $\P$, namely, $\BP$, $\R$, and $\U$, respectively. Unlike $\Sigma_k^\poly$, however, these operators are definitionally identical, modulo a ``witness criterion'' on the number of witness strings.

\begin{definition}
A \emph{witness criterion} $\omega$ is a tuple $(R_1, h_1, R_2, h_2)$, where $h_1, h_2 \in \RR(\NN \rightarrow \NN)$ and $R_1, R_2$ are total computable relations over $\NN$ such as $=, \neq, <, >, \leq, \geq, \equiv_k$, and $\not\equiv_k$.
\end{definition}

Indeed, as we shall see, the following $\f$- and $\omega$-based operator $\W^\f_\omega$ reduces to operators like $\BP$, $\R$, and $\U$, provided that $\f$ and $\omega$ are appropriately specified. Its definition relies on the notion of a \emph{witness set} $W$, where for any function $f$, language $L$, and string $x$, 
$$
W(f, L, x) \defeq \left\{y \in \Sigma^{f(|x|)} \mid \chi_L(\words{x,y}) = 1\right\}.
$$

\begin{definition}
Let $\f$ be an acceptable collection and let $\omega = (R_1, h_1, R_2, h_2)$ be a witness criterion. For a class $\L$ of languages, $\W^\f_\omega \cdot \L$ is the class of languages $L$ for which there exists $f \in \f$ and $L' \in \L$ such that for all $x$,
\begin{equation}
\begin{cases}
x \in L \implies \#W(f, L', x) \,R_1\, h_1(f(|x|)),\\
x \not\in L \implies \#W(f, L', x) \,R_2\, h_2(f(|x|)).
\end{cases}
\label{eq:Woperator}
\end{equation}
\end{definition}

Indeed, $\W^\f_\omega$ reduces to the $\BP$, $\R$, and $\U$ operators:
$$
\BP = \W^{\poly}_{(\geq, 2S/3, \leq, S/3)}, \quad \R = \W^{\poly}_{(\geq, 2S/3, =, 0)}, \quad \text{and} \quad \U = \W^\poly_{(=, 1, =, 0)},
$$
where for all $n$, $S(n) \defeq (\#\Sigma)^n$. In fact, $\W^\f_\omega$ generalizes all the operators in Table~\ref{tab:witnesscriterions}, save $\Sigma_k^\f$ and $\Pi_k^\f$ for $k \geq 2$.

Like $\Sigma_k^\f$, $\W^\f_{\omega}$ is monotone. 
\begin{observation}
For all acceptable collections $\f$, all witness criteria $\omega$, and all language classes $\L_1$ and $\L_2$, $\L_1 \subseteq \L_2$ implies $\W^\f_{\omega} \cdot \L_1 \subseteq \W^\f_{\omega} \cdot \L_2$.
\label{obs:witnessmonotone}
\end{observation}

\begin{table}
\begin{center}
{\renewcommand{\arraystretch}{1.2}
\begin{tabular}{|c|c|c|}
\hline
Complexity Class & Class Operator & Witness Criterion\\
\hline
$\DTIME$ & $\mathsf{ID}$ & $(=, S, =, 0)$\\
\hline
$\BPTIME$ & $\BP$ & $(\geq, 2S/3, \leq, S/3)$\\
\hline
$\PTIME$ & $\P$ & $(>, S/2, <, S/2)$\\
\hline
$\NTIME$ & $\exists^p$ & $(>, 0, =, 0)$\\
\hline
$\co\NTIME$ & $\forall^p$ & $(=, S, <, S)$\\
\hline
$\Sigma_k\TIME$ & $\Sigma_{k}^{\poly}$ & NA for $k \geq 2$\\
\hline
$\Pi_k\TIME$ & $\Pi_{k}^{\poly}$ & NA for $k \geq 2$\\
\hline
$\RTIME$ & $\R$ & $(\geq, S/2, =, 0)$\\
\hline
$\co\RTIME$ & $\co\R$ & $(=, S, \leq, S/2)$\\
\hline
$\UTIME$ & $\U$ & $(=, 1, =, 0)$\\
\hline
$\co\UTIME$ & $\co\U$ & $(=, S - 1, =, S)$\\
\hline
$\parity\TIME$ & $\parity$ & $(\equiv_2, 1, \equiv_2, 0)$\\
\hline
$\co\parity\TIME$ & $\co\parity$ & $(\equiv_2, 0, \equiv_2, 1)$\\
\hline
$\Mod_k\TIME$ & $\Mod_k$ & $(\not\equiv_k, 0, \equiv_k, 0)$\\
\hline
$\co\Mod_k\TIME$ & $\co\Mod_k$ & $(\equiv_k, 0, \not\equiv_k, 0)$\\
\hline
\end{tabular}
}
\end{center}
\caption{The witness criteria that correspond to different witness-based complexity classes and operators. Here, $S(n) \defeq (\#\Sigma)^{n}$ for every $n$.}
\label{tab:witnesscriterions}
\end{table}

Importantly, like $\Sigma_k^\poly \cdot \DTIME(t)$ and $\Sigma_k\TIME(t)$, for general $t$, the classes $\BP \cdot \DTIME(t)$ and $\BPTIME(t)$ are different, as are $\R \cdot \DTIME(t)$ and $\RTIME(t)$, $\U \cdot \DTIME(t)$ and $\UTIME(t)$, and so forth. To illustrate why, we recall the definition of $\BPTIME(t)$.

\begin{definition}
The class $\BPTIME(t)$ consists of all languages $L$ that are decided by probabilistic Turing machines that make polynomially many probabilistic choices and halt in at most $t$ steps. Equivalently, $L \in \BPTIME(t)$ iff there exists $f \in \poly$ and $\varphi_e \in \PPP(\Sigma^* \rightarrow \{0,1\})$ such that for all $x$ and $w$:
\begin{enumerate}[(i)]
\item $\dtime(e, \words{x, w}) \leq t(\words{x})$,
\item and
$$
\begin{cases}
x \in L \implies \#\left\{y \in \Sigma^{f(|x|)} \mid \varphi_e(\words{x,y}) = 1\right\} \geq \frac{2}{3} (\#\Sigma)^{f(|x|)},\\
x \not\in L \implies \#\left\{y \in \Sigma^{f(|x|)} \mid \varphi_e(\words{x,y}) = 1\right\} \leq \frac{1}{3} (\#\Sigma)^{f(|x|)}.
\end{cases}
$$
\end{enumerate}
\end{definition}

Here, condition (i) enforces that $\varphi_e$ is computable on a deterministic Turing machine that halts in at most $t(|\words{x}|)$ steps, as opposed to $t(|\words{x,y}|)$ steps. Thus, in contrast to the definition of the class $\BP \cdot \DTIME(t)$, languages in $\BPTIME(t)$ are decided by deterministic Turing machines whose runtimes are \emph{not} influenced by the witness string $y$. Consequently, for general $t$, it follows from the time hierarchy theorem that $\BP \cdot \DTIME(t) \neq \BP\TIME(t)$. Of course, these same conclusions hold for classes like $\R \cdot \DTIME(t)$ and $\RTIME(t)$, $\U \cdot \DTIME(t)$ and $\UTIME(t)$, and in fact all the classes and their operator-defined counterparts that are listed in Table~\ref{tab:witnesscriterions}.

We now introduce the class $\W_\omega\L_\Phi^\f(t)$, which is defined analogically to the alternate but equivalent definition of $\BPTIME(t)$, but for any acceptable collection $\f$, witness criterion $\omega$, and Blum measure $\Phi$.

\begin{definition}
Let $\f$ be an acceptable collection, let $\Phi$ be a Blum measure, and let $\omega = (R_1, h_1, R_2, h_2)$ be a witness criterion. The class $\W_\omega\L^\f_\Phi(t)$ consists of all languages $L$ for which there exists $f \in \f$ and $\varphi_e \in \PPP(\Sigma^* \rightarrow \{0,1\})$ such that for all $x$ and $w$:
\begin{enumerate}[(i)]
\item $\Phi(e, \words{x, w}) \leq t(\words{x})$,
\item and
$$
\begin{cases}
x \in L \implies \#\left\{y \in \Sigma^{f(|x|)} \mid \varphi_e(\words{x,y}) = 1\right\} R_1\, h_1(f(|x|)),\\
x \not\in L \implies \#\left\{y \in \Sigma^{f(|x|)} \mid \varphi_e(\words{x,y}) = 1\right\} R_2\, h_2(f(|x|)).
\end{cases}
$$
\end{enumerate}
\end{definition}

As above, here condition (i) enforces that $\varphi_e$ is computable in less than $t(|\words{x}|)$ of the Blum measure $\Phi$. In other words, the amount of $\Phi$ it takes to compute $\varphi_e$ is \emph{independent} of the witness string $y$. Again, this is in contrast to the definition of the class $\W_\omega^\f \cdot \L_\Phi(t)$, where the witness string \emph{does} implicitly influence the amount of $\Phi$ it takes to compute $\chi_{L'}$ in \eref{eq:Woperator}. For general $t$, therefore, we expect $\W_\omega^\f \cdot \L_\Phi(t) \neq \W_\omega\L_\Phi^\f(t)$. Nevertheless, as we prove in \sref{sec:mainproofs}, if $t = t_\f$ in \propref{prop:mainprop1}, then equality does in fact hold.

\section{Statement of Main Result}

We now state our main theorem.

\begin{customthm}{A}
\label{thm:uniontheorem}
Let $\f$ be an acceptable collection and let $\{\Phi_i \mid i \in [\ell]\}$ be a finite number of Blum measures. There exists a non-decreasing function $t_\f \in \RR(\NN \rightarrow \NN)$ such that for all $i \in [\ell]$, all $k$, and all witness criteria $\omega$:
\begin{enumerate}[(i)]
\item $\L_{\Phi_i}(t_\f) = \bigcup_{f \in \f}\L_{\Phi_i}(f)$,
\item $\Sigma_k\L^{\f}_{\Phi_i}(t_\f) = \bigcup_{f \in \f} \Sigma_k\L^{\f}_{\Phi_i}(f)$,
\item $\W_{\omega}\L^\f_{\Phi_i}(t_\f) = \bigcup_{f \in \f}\W_{\omega}\L^\f_{\Phi_i}(f)$.
\end{enumerate}
\end{customthm}

Note how this establishes an \emph{infinity} of complexity classes that are all characterized by the \emph{same} bound $t_\f$. We regard this as evidence for \conjref{conj:mainconj}. Importantly, \thmref{thm:uniontheorem} relativizes.

Notice, with $\f = \poly$ and just the single Blum measure $\dtime$, \thmref{thm:uniontheorem} yields our main corollary, which is an even stronger version of \corref{cor:intcorollary}.

\begin{customcor}{A}
\label{cor:maincor}
There exists a non-decreasing function $t_\poly \in \RR(\NN \rightarrow \NN)$ such that for all $k$ and all witness criteria $\omega$:
\begin{enumerate}[(i)]
\item $\P = \DTIME(t_\poly)$,
\item $\NP = \NTIME(t_\poly)$,
\item $\coNP = \co\NTIME(t_\poly)$,
\item $\PP = \PTIME(t_\poly)$,
\item $\BPP = \BPTIME(t_\poly)$,
\item $\RP = \RTIME(t_\poly)$,
\item $\co\RP = \co\RTIME(t_\poly)$,
\item $\UP = \UTIME(t_\poly)$,
\item $\co\UP = \co\UTIME(t_\poly)$,
\item $\PSPACE = \DSPACE(t_\poly)$,
\item $\Mod_k\P = \Mod_k\TIME(t_\poly)$,
\item $\co\Mod_k\P = \co\Mod_k\TIME(t_\poly)$,
\item $\Sigma_k\P = \Sigma_k\TIME(t_\poly)$,
\item $\Pi_k\P = \Pi_k\TIME(t_\poly)$,
\item $\W_\omega\L_\dtime^\poly(t_\poly) = \bigcup_{f \in \poly}\W_\omega\L_\dtime^\poly(f)$.
\end{enumerate}
\end{customcor}

Of course, save statement (xv), \corref{cor:maincor} more or less follows from our informal discussion in \sref{sec:intro} and the fact that $t_\poly$ is non-decreasing. Nevertheless, it also follows from our more general result, whose proof, as we now show, is a bit more involved.

\section{Proof of Main Result}
\label{sec:mainproofs}

We start by proving several things about the $\Sigma_{k}^\f$ and $\W_{\omega}^\f$ operators.

\begin{lemma}
Let $\f$ be an acceptable collection and let $\Phi$ be a Blum measure. For all $f \in \f$, all $k$, and all witness criteria $\omega$, there exist $g,h \in \f$ such that:
\begin{enumerate}[(i)]
\item $\Sigma_{k}^{\f} \cdot \L_\Phi(f) \subseteq \Sigma_{k}\L_\Phi^{\f}(g) \subseteq \Sigma_{k}^{\f} \cdot \L_\Phi(g)$,
\item $\W_{\omega}^{\f} \cdot \L_\Phi(f) \subseteq \W_{\omega}\L_\Phi^{\f}(h) \subseteq \W_{\omega}^{\f} \cdot \L_\Phi(h)$.
\end{enumerate}
\label{lem:subsetlem}
\end{lemma}
\begin{proof}
The proof of (i) is below; the proof of (ii) is analogous.

Fix $f \in \f$ and let $L \in \Sigma_k^\f \cdot \L_\Phi(f).$ Then, there are $f_1, \dots, f_k \in \f$ and $L' \in \L_\Phi(f)$ such that for all $x$,
$$
x \in L \iff \left(\exists y_1 \in \Sigma^{f_1(|x|)}\right)\left(\forall y_2 \in \Sigma^{f_2(|x|)}\right)\dots\left(Q_k y_k \in \Sigma^{f_k(|x|)}\right) \chi_{L'}(\words{x, y_1, \dots, y_k}) = 1.
$$
Since $L' \in \L_\Phi(f)$, for all $x, w_1, \dots, w_k$,
$$
\Phi(e, \words{x, w_1, \dots, w_k}) \leq f(|\words{x, w_1, \dots, w_k}|),
$$
where $e$ is the G\"odel number of $\chi_{L'}$. To establish the first containment in (i), it remains to show that there is $g \in \f$ such that for all $x, w_1, \dots, w_k$, $\Phi(e, \words{x, w_1, \dots, w_k}) \leq g(|\words{x}|)$. To this end, note that by \eref{eq:pairingfunctionsize},
\begin{align*}
|\words{x, w_1, \dots, w_k}| &\leq |x| + \sum_{i = 1}^k |w_i| + O\left(\log |x| + \sum_{i = 1}^{k - 1} \log |w_i|\right)\\
&= |x| + \sum_{i = 1}^k f_i(|x|) + O\left(\log |x| + \sum_{i = 1}^{k - 1} \log f_i(|x|)\right).
\end{align*}
Therefore, 
$$
\Phi(e, \words{x, w_1, \dots, w_k}) \leq f\left(|x| + \sum_{i = 1}^k f_i(|x|) + O\left(\log |x| + \sum_{i = 1}^{k - 1} \log f_i(|x|)\right)\right).
$$
Thus, by property (e) of the acceptable collection $\f$, there is indeed $g \in \f$ such that for all $x, w_1, \dots, w_k$, $\Phi(e, \words{x, w_1, \dots, w_k}) \leq g(|x|) \leq g(|\words{x}|)$, where the latter inequality holds because $|x| \leq |\words{x}|$ and $g$ is non-decreasing (property (c) of the acceptable collection $\f$). Consequently, $L \in \Sigma_k\L_\Phi^\f(g)$, as desired.

We now establish the second containment in (i). Let $L \in \Sigma_k\L_{\Phi}(g).$ Then there exist $g_1, \dots, g_k \in \f$ and $\varphi_e \in \PPP(\Sigma^* \rightarrow \{0,1\})$ such that for all $x,w_1, \dots, w_k$, $\Phi(e, \words{x, w_1, \dots, w_k}) \leq g(|\words{x}|)$ and
$$
x \in L \iff \left(\exists y_1 \in \Sigma^{g_1(|x|)}\right)\left(\forall y_2 \in \Sigma^{g_2(|x|)}\right)\dots\left(Q_k y_k \in \Sigma^{g_k(|x|)}\right) \varphi_e(\words{x, y_1, \dots, y_k}) = 1.
$$
Since $g$ is non-decreasing,
$$
\Phi(e, \words{x, y_1, \dots, y_k}) \leq g(|\words{x}|) \leq g(|\words{x, y_1, \dots, y_k}|).
$$ 
Therefore, $L' \defeq \{x \mid \varphi_e(x) = 1\} \in \L_{\Phi}(g)$ and
$$
x \in L \iff \left(\exists y_1 \in \Sigma^{g_1(|x|)}\right)\left(\forall y_2 \in \Sigma^{g_2(|x|)}\right)\dots\left(Q_k y_k \in \Sigma^{g_k(|x|)}\right) \chi_{L'}(\words{x, y_1, \dots, y_k}) = 1.
$$
Thus, $L \in \Sigma_k^\f \cdot \L_{\Phi_i}(g)$, as desired.
\end{proof}

The next lemma shows that the $\Sigma^\f_{k}$ and $\W_\omega^\f$ operators pass through unions over $\f$.

\begin{lemma}
Let $\f$ be an acceptable collection and let $\Phi$ be a Blum measure. For all $k$ and all witness criteria $\omega$:
\begin{enumerate}[(i)]
\item $\Sigma_k^\f \cdot \bigcup_{f \in \f} \L_\Phi(f) = \bigcup_{f \in \f} \Sigma_k^\f \cdot \L_\Phi(f)$,
\item $\W_\omega^\f \cdot \bigcup_{f \in \f} \L_\Phi(f) = \bigcup_{f \in \f} \W_\omega^\f \cdot \L_\Phi(f)$.
\end{enumerate}
\label{lem:unionPass}
\end{lemma}
\begin{proof}
The proof of (i) is below; the proof of (ii) is analogous.

Since $\L_\Phi(f') \subseteq \bigcup_{f \in \f} \L_\Phi(f)$ for all $f' \in \f$, the monotonicity of $\Sigma_k^\f$ (\obsref{obs:monotone}) implies $\Sigma_k^\f \cdot \L_\Phi(f') \subseteq \Sigma_k^\f \cdot \bigcup_{f \in \f} \L_\Phi(f)$. Hence,
\begin{align*}
\bigcup_{f' \in \f}\Sigma^\f_k \cdot \L_\Phi(f') &\subseteq \bigcup_{f' \in \f}\Sigma^\f_k \cdot \bigcup_{f \in \f} \L_\Phi(f)\\ 
& = \Sigma^\f_k \cdot \bigcup_{f \in \f} \L_\Phi(f).
\end{align*}

For the other direction, let $L \in \Sigma^\f_{k} \cdot \bigcup_{f \in \f} \L_{\Phi}(f)$. Then, there are $f_1, \dots, f_k \in \f$ and $L' \in \bigcup_{f \in \f} \L_{\Phi}(f)$ such that for all $x$, \eref{eq:operatoraction} holds. But $L' \in \bigcup_{f \in \f} \L_{\Phi}(f)$ implies $L' \in \L_{\Phi}(f')$ for some $f' \in \f$. Therefore, $L \in \Sigma^\f_{k} \cdot \L_\Phi(f') \subseteq \bigcup_{f \in \f}\Sigma^\f_{k} \cdot \L_\Phi(f)$.
\end{proof}

The rest of our results rely heavily on the following proposition, which is a weaker version of \propref{prop:mainprop1}.

\begin{proposition}
Let $\{\Phi_i \mid i \in [\ell] \}$ be a finite set of Blum measures and let $\f$ be an acceptable collection. There exists $t_\f \in \RR(\NN \rightarrow \NN)$ such that:
\begin{enumerate}[(i)]
\item $t_\f$ is non-decreasing,
\item for all $f \in \f$, $f \leq t_\f$,
\item for all $i \in [\ell]$, $\L_{\Phi_i}(t_\f) = \bigcup_{f \in \f}\L_{\Phi_i}(f)$.
\end{enumerate}
\label{prop:mainprop}
\end{proposition}

\begin{lemma}
Let $\f$ be an acceptable collection, let $\{\Phi_i \mid i \in [\ell]\}$ be a finite collection of Blum measures, and let $t_\f$ be as in \propref{prop:mainprop}. For all $i \in [\ell]$, all $k$, and all witness criteria $\omega$:
\begin{enumerate}[(i)]
\item $\L_{\Phi_i}(t_\f) = \bigcup_{f \in \f}\L_{\Phi_i}(f)$,
\item $\Sigma^{\f}_{k} \cdot \L_{\Phi_i}(t_\f) = \bigcup_{f \in \f} \Sigma_k\L_{\Phi_i}^{\f}(f),$
\item $\W^{\f}_{\omega} \cdot \L_{\Phi_i}(t_\f) = \bigcup_{f \in \f} \W_\omega\L_{\Phi_i}^{\f}(f).$
\end{enumerate}
\label{lem:mainlemone}
\end{lemma}

\begin{proof}
\propref{prop:mainprop} proves (i). The proof of (ii) is below; the proof of (iii) is analogous.

The following is true for all $i \in [\ell]$ and all $k$. By \lemref{lem:unionPass} and the definition of $t_\f$, $L \in \Sigma^{\f}_{k} \cdot \L_{\Phi_i}(t_\f)$ iff $L \in \bigcup_{f \in \f} \Sigma^\f_{k} \cdot \L_{\Phi_i}(f)$ iff $L \in \Sigma^\f_k \cdot \L_{\Phi_i}(f')$ for some $f' \in \f$. By \lemref{lem:subsetlem}, there is then $g \in \f$ such that $\Sigma^\f_k \cdot \L_{\Phi_i}(f') \subseteq \Sigma_k\L_{\Phi_i}^\f(g) \subseteq \Sigma^\f_k \cdot \L_{\Phi_i}(g)$, which implies 
$$
\bigcup_{f \in \f} \Sigma^\f_{k} \cdot \L_{\Phi_i}(f) = \bigcup_{f \in \f} \Sigma_k\L_{\Phi_i}^\f(f).
$$
The result then follows from \lemref{lem:unionPass} and \propref{prop:mainprop}.
\end{proof}

\begin{lemma}
Let $\f$ be an acceptable collection, let $\{\Phi_i \mid i \in [\ell]\}$ be a finite collection of Blum measures, and let $t_\f$ be as in \propref{prop:mainprop}. For all $i \in [\ell]$, all $k$, and all witness criteria $\omega$:
\begin{enumerate}[(i)]
\item $\L_{\Phi_i}(t_\f) = \bigcup_{f \in \f}\L_{\Phi_i}(f)$,
\item $\Sigma^\f_{k} \cdot \L_{\Phi_i}(t_\f) = \Sigma_k\L_{\Phi_i}^\f(t_\f)$,
\item $\W^\f_{\omega} \cdot \L_{\Phi_i}(t_\f) = \W_\omega\L_{\Phi_i}^\f(t_\f)$. 
\end{enumerate}
\label{lem:mainlemtwo}
\end{lemma}
\begin{proof}
\propref{prop:mainprop} proves (i). The proof of (ii) is below; the proof of (iii) is analogous.

The following is true for all $i \in [\ell]$ and all $k$. First, let $L \in \Sigma^\f_k \cdot \L_{\Phi_i}(t_\f)$. Then, by \lemref{lem:mainlemone}, $L \in \Sigma_k\L_{\Phi_i}^\f(f')$ for some $f' \in \f$. But $f' \leq t_\f$ by \propref{prop:mainprop}, so $L \in \Sigma_k\L_{\Phi_i}^\f(t_\f)$. 

Now let $L \in \Sigma_k\L_{\Phi_i}(t_\f).$ Then, there exist $f_1, \dots, f_k \in \f$ and $\varphi_e \in \PPP(\Sigma^* \rightarrow \{0,1\})$ such that for all $x,w_1, \dots, w_k$, $\Phi_i(e, \words{x, w_1, \dots, w_k}) \leq t_\f(|\words{x}|)$ and
$$
x \in L \iff \left(\exists y_1 \in \Sigma^{f_1(|x|)}\right)\left(\forall y_2 \in \Sigma^{f_2(|x|)}\right)\dots\left(Q_k y_k \in \Sigma^{f_k(|x|)}\right) \varphi_e(\words{x, y_1, \dots, y_k}) = 1.
$$
Since $t_\f$ is non-decreasing (\propref{prop:mainprop}),
$$
\Phi_i(e, \words{x, y_1, \dots, y_k}) \leq t_\f(|\words{x}|) \leq t_\f(|\words{x, y_1, \dots, y_k}|).
$$ 
Therefore, $L' \defeq \{x \mid \varphi_e(x) = 1\} \in \L_{\Phi_i}(t_\f)$ and
$$
x \in L \iff \left(\exists y_1 \in \Sigma^{f_1(|x|)}\right)\left(\forall y_2 \in \Sigma^{f_2(|x|)}\right)\dots\left(Q_k y_k \in \Sigma^{f_k(|x|)}\right) \chi_{L'}(\words{x, y_1, \dots, y_k}) = 1.
$$
Thus, $L \in \Sigma_k^\f \cdot \L_{\Phi_i}(t_\f)$, as desired.
\end{proof}

Together Lemmas~\ref{lem:mainlemone} and \ref{lem:mainlemtwo} imply \thmref{thm:uniontheorem}.

\section*{Acknowledgements}

The authors thank Scott Aaronson, Lance Fortnow, Josh Grochow, and the anonymous referees at CCC'24 for substantive comments on an earlier draft of this paper.


\bibliographystyle{amsplain}
\bibliography{references}

\end{document}